\documentclass{article}[11pt,a4paper]

\usepackage[utf8]{inputenc}
\usepackage{booktabs}

\usepackage{amsthm}
\usepackage{amsmath}
\usepackage{amsfonts}
\usepackage{authblk}

\newtheorem{theorem}{Theorem}[section]
\newtheorem{lemma}[theorem]{Lemma}
\newtheorem{corollary}[theorem]{Corollary}

\usepackage{complexity}
\usepackage{cite}
\usepackage{xspace}
\usepackage[usenames,dvipsnames]{xcolor}
\usepackage{todonotes}%disable
\usepackage{caption}
\usepackage{subcaption}
\usepackage{float}
\usepackage{url}
\usepackage{paralist}
\usepackage{nicefrac}
\usepackage{eqnarray}
\usepackage[margin=1.5in]{geometry}

\newcommand{\EDP}{EDP\xspace}	
	
\newcommand{\net}{\ensuremath{\mathrm{net}}}
\newcommand{\plotscale}{0.8}
\newcommand{\plotscaleOne}{0.45}
\begin{document}
\author[1]{Thomas Leibfried}
\author[2]{Tamara Mchedlidze}
\author[1]{Nico Meyer-H\"ubner}
\author[2]{\mbox{Martin N\"ollenburg}}
\author[2]{Ignaz Rutter}
\author[2]{Peter Sanders}
\author[2]{\authorcr Dorothea Wagner}
\author[2]{Franziska Wegner}

\affil[1]{Institute of Energy Systems and High Voltage Technology}
\affil[2]{Institute of Theoretical Informatics\authorcr \vspace{1ex} Karlsruhe Institute of Technology, Germany}

% \author{
% Thomas Leibfried\thanks{Institute of Energy Systems and High Voltage Technology, Karlsruhe Institute of Technology, Germany}
% \and
% Tamara Mchedlidze\thanks{Institute of Theoretical Informatics, Karlsruhe Institute of Technology, Germany.}
% \and
% Nico Meyer-H\"ubner$^*$ %$\overset{*}{,}$
% \and
% Martin N\"ollenburg$^\dag$ %$\overset{\dag}{,}$
% \and
% Ignaz Rutter$^\dag$ %$\overset{\dag}{,}$
% \and
% Peter Sanders$^\dag$ %$\overset{\dag}{,}$
% \and
% Dorothea Wagner$^\dag$ %$\overset{\dag}{,}$
% \and
% Franziska Wegner$^\dag$ %$\overset{\dag}{\vphantom{,}}$
% }

\title{Operating Power Grids with Few Flow Control Buses} 
\date{}
\maketitle
\begin{abstract}
Future power grids will offer enhanced controllability due to the increased availability of power flow control units (FACTS). As the installation of control units in the grid is an expensive investment, we are interested in using few controllers to achieve high controllability. In particular, two questions arise: How many flow control buses are necessary to obtain globally optimal power flows? And if fewer flow control buses are available, what can we achieve with them?

Using steady state IEEE benchmark data sets, we explore experimentally that already a small number of controllers placed at certain grid buses suffices to achieve globally optimal power flows. We present a graph-theoretic explanation for this behavior. To answer the second question we perform a set of experiments that explore the existence and costs of feasible power flow solutions at increased loads with respect to the number of flow control buses in the grid. We observe that adding a small number of flow control buses reduces the flow costs and extends the existence of feasible solutions at increased load.
\end{abstract}
\section{Introduction}
\label{sec:introduction}
The central task of any electrical power infrastructure is the reliable and cost-efficient supply of electrical energy to industry and population on a national or even continental scale. 
Future power grids and their usage are subject to fundamental changes due to the shift towards renewable distributed energy production and the installation of new power flow control units, which offer increased control, but make the grid operation more demanding.
Not only do these changes lead to a much larger number of independent power producers (IPP), which are highly distributed in the network, but they also cause very different patterns of energy flow. 
For example, regions with off-shore wind farms may sometimes produce enough energy to supply remote consumers, but at other times they are consumers themselves.
In particular, this may require long-distance energy transmission and frequent flow direction changes. Most of the existing power grids, however, were not designed for such transmission patterns. The current strategy to cope with these changes is to either extend the grid with additional transmission lines, or to install advanced control units to facilitate better utilization of the existing infrastructure.

In this paper, we consider the latter option and study the advantageous effects of making selected buses of a power grid to controllable, both in terms of the minimum number of controllable buses needed for achieving maximum flow control and in terms of the operation costs and the existence of feasible power flows at critical line capacities. 

In abstract terms, we assume that a \emph{flow-control bus} is able to flexibly distribute the entire power flow at this bus among the incident edges, as long as Kirchhoff's current law (or the \emph{flow-conservation property}) is satisfied, i.e., the in-flow to the bus equals its out-flow.
These flow control buses can be realized using power electronics devices known as \emph{flexible AC transmission systems} (FACTS), which are a class of power systems that have the capabilities to control various electrical bus parameters~\cite{Hingorani1993a,gh-ufact-00}. 
More specifically, since we are interested in controlling the real power flow on the branches incident to a particular bus, we can realize our flow control buses by installing on each (but one) incident branch a \emph{unified power flow controller} (UPFC), which is a FACTS that is able to control the voltage magnitude and angle and consequently has control of the real and reactive power flow on the particular branch~\cite{Noroozian1997,gh-ufact-00}.

One of the most important tasks in operating a power grid is to decide the energy production of each power generator such that supply and consumption are balanced and the resulting power flow does not exceed the thermal limits of the power lines. 
Among all solutions we are interested in one that minimizes the total energy production and transmission cost. 
This is called \emph{Economic Dispatch Problem} (EDP).
The main approach for solving this problem in power grids without FACTS is the optimal power flow (OPF) method, a numerical method that was introduced in 1962 by Carpentier~\cite{Carpentier62} and has subsequently been refined and generalized, see the recent surveys by Frank et al.~\cite{SurveyOPF1,SurveyOPF2}. 
However, OPF is not designed for hybrid power grids with flow control buses and cannot exploit the extended flow control possibilities to obtain globally optimal solutions. 

Hence, we propose in Section~\ref{sec:model} a new hybrid DC-based model for power flows in power grids that combine traditional grid buses with some flow-control buses.
In order to answer our questions on the effects of installing flow control buses, we solve the \EDP in our hybrid model using a linear programming (LP) formulation.
Our LP combines a standard graph-theoretical network flow model, which already includes Kirchhoff's current law at all buses, with additional constraints for Kirchhoff's voltage law in those parts of the grid that are not equipped with flow control buses. Thus we are able to obtain electrically feasible power flows that minimize, similarly to OPF, the overall flow costs in terms of generation and transmission costs.

Using the well-known IEEE power systems test cases, we performed simulation experiments related to two key questions, which take into account that the FACTS needed for realizing our flow control buses in reality constitute a significant and expensive investment and hence their number should be as small as possible. 
\begin{compactenum}
	\item How many flow control buses are necessary to obtain globally optimal power flows and which buses need to be controlled?
	\item If the number of available flow control buses is given, do we still see a positive effect on the flow costs and on the operability of the grid when approaching its capacity limits?
\end{compactenum}

In Section~\ref{sec:hybridmodel} we address the first question. 
In our experiments we determine the minimum number of flow control buses necessary to achieve the same solution quality as in a power grid in which each bus is controllable and which clearly admits an upper bound on what can be achieved with the network topology. 
Interestingly, it turns out that a relatively small number of flow control buses are sufficient for this.
In fact, we can prove a theorem stating a structural graph-theoretic property, which, if met by the placement of flow control buses, implies the optimality of the power flow and serves as a theoretical explanation of the observed behavior.
Section~\ref{sec:grid-control-when-approaching-capacity-limits} deals with the second question of operating a power grid close to its capacity limits, which becomes increasingly relevant as the consumption of electrical energy grows faster than the grid capacities.
Our experiments indicate that installing few flow control buses in a power grid is sufficient not only to achieve  lower costs compared to an OPF solution, but also allows to operate the grid at capacities for which no feasible OPF solution exists any more.
\section{Related Work}
\label{sec:related-work}
With the increasing availability and technological advancement of
FACTS researchers began to study the possible benefits of their
installation in power grids from different perspectives.

From an economic perspective, it is of interest to support investment
decisions in power grid expansion planning by considering alternative
investment strategies that either focus on new transmission lines or
allow mixed approaches including FACTS placement.  Blanco et
al. \cite{bogr-rovfi-11} present a least-squares Monte-Carlo method
for evaluating investment strategies and argue that FACTS allow for a
more flexible, mixed strategy that fares better under uncertainty.
Tee and Ilic present an optimal decision-making framework for
comparing investment decisions, including FACTS~\cite{ti-oidte-12}.

From the perspective of operating a power grid, the main question is
how many and where FACTS should be placed in order to optimize a
certain criterion.  Cai et al.~\cite{1397562} propose and
experimentally evaluate a genetic algorithm for allocating different
types of FACTS in a power grid in order to optimally support a
deregulated energy market.  Gerbex et al.~\cite{gcg-olmtf-01} and
Ongsakul and Jirapong~\cite{1465551} study the placement of FACTS with
the goal of increasing the amount of energy that can be transferred.
Gerbex et al.~\cite{gcg-olmtf-01} present a genetic algorithm that
optimizes simultaneously the energy generation costs, transmission
losses, line overload, and the acquisition costs for FACTS.  Ongsakul
and Jirapong~\cite{1465551} use evolutionary programming to place
FACTS such that the total amount of energy that can be transferred
from producers to consumers is maximized; in contrast to our setting,
they may also increase the demands of consumers arbitrarily.
In contrast to these heuristic approaches Lima et
al.~\cite{lgkm-psp-03} use mixed-integer linear programming to
optimally increase the loadability of a system by placing FACTS
subject to limits on their number or cost.  Similar to our approach,
they do not distinguish different types of FACTS but rather assume
``ideal'' FACTS that can control all transmission parameters of a
branch.  In contrast to our work, they focus only on loadability and
do not consider generation costs and line losses.

All related work mentioned so far considers the DC model for
electrical networks as an approximation to the AC model and aims at
providing a preliminary step in an actual planning process, where this
approximation is sufficient.  There are also a few attempts to solve
the placement problem for FACTS in the more realistic but also more
complicated AC model.  
Sharma et al.~\cite{sgv-nps-03} develop an
evaluation whether transmission lines are critical and propose to
place FACTS at critical lines in order to improve voltage stability in
the grid.  Ippolito and Siano~\cite{is-sonl-04} present a genetic
algorithm for FACTS placement in AC networks and experimentally
evaluate it in a case study.
In contrast to these heuristic approaches, Farivar and Low~\cite{6507352} 
observe exact OPF evaluation in a relaxed AC-model. In this context, they 
place phase shifters to exploit structural characteristics that are similar to our approach.
\section{Preliminaries}
\label{sec:preliminaries}
In this section we recall some basic notions from graph theory.
Although, for technical reasons, the graphs we use for modeling power
grids are directed, when considering the topology of the network, we
always consider the underlying undirected graph.  Thus, in the
following let $G$ be an undirected graph.

The graph $G$ is \emph{connected} if it contains a path between any
two vertices.  A \emph{connected component} of $G$ is a maximal
connected subgraph of $G$ (maximal with respect to inclusion).  A
\emph{cactus} is a graph where every edge is contained in at most one
cycle.  A \emph{forest} is a graph that does not contain a cycle.

A \emph{cutvertex} is a vertex of a graph whose removal increases the
number of connected components.  A \emph{biconnected component} is a
maximal subgraph that does not have a cutvertex.  Note that a
biconnected component of a forest is either \emph{trivial} in the
sense that it consists of a single vertex, or it consists of a single
edge.  Similarly, a biconnected component of a cactus is trivial, a
single edge, or a cycle.

A \emph{feedback set} of $G = (V,E)$ with respect to a class of graphs
$\mathcal G$ is a set of vertices $F \subseteq V$ such that $G-F \in
\mathcal G$.  We will only be interested in feedback sets with respect
to forests and cacti.  The former is also called \emph{feedback vertex
  set}.  Naturally, one is interested in finding a set $F$ that is as
small as possible.
\section{Model}
\label{sec:model}
In this section we introduce three graph-theoretic flow models for optimal power flows.
Our models are based on the DC power grid model~\cite{Hammerstrom2007, Zimmerman2011a,sja-dcpfr-09}, which is commonly used as an approximation of AC grids~\cite{Purchala2005,ocs-acdc-04}.
We model a power grid as a graph $G = (V, E)$, where $V$ is the set of vertices and $E\subseteq \binom{V}{2}$ is the set of edges. The vertices represent the buses, some of which may be special generator and consumer buses, and the edges represent the branches, which may be transmission lines between the incident buses or transformers.
There is a subset $V_G \subseteq V$ of the vertices that represents generator buses.
Each generator $g \in V_G$ has a maximum supply $x_g \in \mathbb R^{+}$ and is equipped with a convex cost function $\gamma_g > 0$ that is assumed to be piecewise linear with
\begin{equation}
	\label{eq:generatorcostfunction}	
	\gamma_g(x) = \max\{a_i x + c_i \mid(a_i, c_i)\in F_g\},
\end{equation}
where $F_g$ is the set of all piecewise linear functions of $\gamma_g$ and $a_i\leq a_{i+1}$. 
Further, there is a subset $V_C \subseteq V \setminus V_G$ of consumer buses.
Each consumer $u \in V_C$ has a power demand $d_u \in \mathbb R$.

Each branch $e \in E$ has a thermal limit, which is modeled as a capacity function $c\colon E\to\mathbb{R}$ restricting the real power. 
Further, each branch causes a certain loss of power depending on the physical branch parameters and the actual power flow on the branch. 
%\todo{Anders formulieren}
These losses are again approximated as a convex, piecewise linear function $\ell_{e}$ for each edge $e \in E$ with
\begin{equation}
	\label{eq:losscostfunction}	
	\ell_{e}(x) = \max\{a_i x + c_i \mid(a_i, c_i)\in F_{e}\},
\end{equation}
where $F_e$ is the set of all piecewise linear functions of $\ell_e$ and $a_i\leq a_{i+1}$.

A \emph{flow} $f$ in the power grid $G$ is a function $f \colon V \times V \to
\mathbb{R}$ with the property that 
\begin{equation}
	f(u,v) = -f(v,u) \quad \forall u,v \in V
\end{equation}
For every vertex $u$ in $G$, we define its \emph{net
  out-flow} \[f_\net(u) = \sum_{\{u,v\} \in E} f(u,v).\] 
For a flow $f$, we further define two types of costs, the \emph{generator
  costs} \[c_g(f) = \sum_{g \in V_G} \gamma_g(f_\net(g))\] and the
\emph{line losses} \[c_\ell(f) = \sum_{\{u,v\} \in E} \ell_{\{u,v\}}(|f(u,v)|)\,.\] To
obtain the overall cost for the flow $f$, we weight these two terms as
\begin{equation}
c_\lambda(f) = \lambda \cdot c_g(f) + (1-\lambda) \cdot c_\ell(f)\label{eq:objective}
\end{equation}
where $\lambda\in [0,1]$.  Our goal is to minimize this objective
function in several different power flow models.
\subsection{Power Flow Models}
\label{sec:power-flow-models}
The most basic model is the \emph{flow model}, where $f$ has to
satisfy the following constraints.
\begin{align}
 -c(e) \le f(u,v) & \le c(e)  &  \forall e=\{u,v\} \in E\label{eq:capacity}\\
f_\net(v) & =  0 & v \in V\setminus (V_G \cup V_C)\label{eq:conservation}\\
f_\net(v) & =   -d_v & v \in V_C\label{eq:consumer}\\
0 \le f_\net(v) & \le x_v &  v \in V_G\label{eq:generator} 
\end{align}
We call a flow satisfying these constraints \emph{feasible}.
Equation~(\ref{eq:capacity}) models the thermal limits or real power capacities of all branches and is called \emph{capacity constraint}. 
Equation~(\ref{eq:conservation}) models that vertices that are neither generators nor consumers have zero net out-flow and is called \emph{flow conservation constraint}. 
Equation~(\ref{eq:consumer}) models that all consumer demands are satisfied and is called \emph{consumer constraint}.
Finally, Equation~(\ref{eq:generator}) models that all generators respect their production limits and is called \emph{generator constraint}.
%Note that vertices with net out-flow 0 satisfy Kirchhoff's current law.

The flow model neglects some physical properties of electrical flows,
in particular Kirchhoff's voltage law.  Thus, the computed power
flows can only be applied to power grids where every vertex is
a control vertex.  In contrast, the \emph{electrical flow
  model}, e.g., according to Zimmerman et al.~\cite{Zimmerman2011a}, models the
power flow via the same set of constraints as the flow model, but additionally requires
the existence of a suitable voltage angle assignment $\Theta \colon V \to
\mathbb{R}$ such that for each branch $\{u,v\}$ the following equation
holds 
\begin{equation}
  \label{eq:electricalflowintro}
  f(u,v) = B(u,v) (\Theta(u) - \Theta(v))\, .  
\end{equation}
Here $B(u,v)$ is the \emph{susceptance} of the branch $(u,v)$.  This is
equivalent to restricting the model to feasible flows that also satisfy
Kirchhoff's voltage law, or, in other words, no flow control buses are used.
This yields a model that matches the situation in the traditional power grids existing today.
We call a feasible flow $f$ \emph{electrically feasible} if
there exists a voltage angle assignment $\Theta$
satisfying~(\ref{eq:electricalflowintro}).
% \franzi[Nico's suggestion]{This model can be easily extended by phase shifters including phase shift angles $\Theta_{\mathrm{shift}}$~\cite{Zimmerman2011}.}

Recall from the introduction that flow control buses can be technically realized by \mbox{UPFCs}, which is a FACTS. Ideal FACTS as introduced by Griffin et al.~\cite{julieGriffin} are often used to simplify the modeling of FACTS by using a linear model and assuming a complete and independent control of the real and reactive power. Our flow control buses are ideal FACTS that control the power flow to all incident edges. 
The flow model---in contrast to the electrical model---assumes flow control buses at each vertex, whereas the
electrical model assumes no immediate control of the power flow.  
Instead, the grid is balanced by changing
the generator outputs only.  In the following we propose a
\emph{hybrid model} that combines the flow model and the electrical flow model in order to handle power grids with flow control buses at a subset of selected vertices.

Let $F \subseteq V$ be a subset of vertices of $G$.  We denote by
$G_F$ the power network obtained from $G$ by considering all vertices in
$F$ as flow control buses.  We call any subgraph~$G' = G[V']$ induced by
a subset $V' \subseteq V \setminus F$ of the vertices without controllers a \emph{native power grid} of $G$.  A flow of $G$ is
\emph{electrically feasible} for a native power grid $G' \subseteq G$
if there exists a voltage angle assignment $\Theta \colon E \to \mathbb{R}$
such that every edge in $G'$ satisfies
Equation~(\ref{eq:electricalflowintro}).  In this case we call
$\Theta$ \emph{feasible (voltage) angle assignment} for $G'$.

A feasible flow $f$ is \emph{electrically feasible for $G_F$} if and
only if $f$ is electrically feasible for the \emph{maximal native
  power grid} $G-F = G[V \setminus F]$.  Intuitively, this models the
fact that a power flow in $G_F$ must be a feasible flow and that it
satisfies the second Kirchhoff law in the maximum native power grid.

Obviously, if $F \subseteq F'$ and $f$ is an electrically feasible
flow for $G_F$, then $f$ is also electrically feasible for $G_{F'}$.
Hence the minimum value of the cost $c_\lambda$ does not increase when adding
more flow control buses.

We note that each of the models can easily be expressed as a linear
program (LP), and thus in all three models an optimal solution can be
computed efficiently~\cite{Bazaraa:2004:LPN:1062374}; see Appendix~\ref{app:lp}. 
However, the flow model can be reduced to a special minimum cost network flow 
problem, for which efficient exact optimization algorithms exist~\cite{Goldberg19971}.  We 
describe this reduction in the following.
%
% latex table generated in R 3.1.2 by xtable 1.7-4 package
% Mon Jan 12 00:45:02 2015
\begin{table}[tb!]
\centering
\begin{tabular}{lrrrr}
  \toprule
  \multicolumn{1}{c}{case} & \multicolumn{1}{c}{$nb$} & \multicolumn{1}{c}{$nl$} & \multicolumn{1}{c}{$ng$} & \multicolumn{1}{c}{$p_d$} \\ 
  \midrule
  \texttt{case6} &   6 &  11 &   3 & 210.00 \\ 
  \texttt{case9} &   9 &   9 &   3 & 315.00 \\ 
  \texttt{case14} &  14 &  20 &   5 & 259.00 \\ 
  \texttt{case30} &  30 &  41 &   6 & 189.20 \\ 
  \texttt{case39} &  39 &  46 &  10 & 6254.23 \\ 
  \texttt{case57} &  57 &  78 &   7 & 1250.80 \\ 
  \texttt{case118} & 118 & 179 &  54 & 4242.00 \\ 
   \bottomrule
\end{tabular}
\caption{IEEE benchmark set with $nb$, $nl$, $ng$ and $p_d$ representing number of buses, number of transmission lines, number of generators and total power demand, respectively.} 
\label{tab:examples}
\end{table}
\subsection{Reduction to MinCostFlow}
\label{sec:model-mincost-flow}
Let  $N = (G = (V,E); (s,t); c; a)$ be an \emph{$s$-$t$ flow network}
% A minimum-cost $s$-$t$ network flow model is an $s$-$t$ flow network $N = (G = (V,E); (s,t); c; a)$
consisting of a directed \mbox{(multi-)} graph $G$, two dedicated source and sink vertices $s, t \in V$, edge capacities $c: E \to \mathbb R_0^+$, and edge costs $a: E \to \mathbb R_0^+$. 
A \emph{flow} $f$ in $N$ is a function $f: E \to \mathbb R_0^+$ and it is called \emph{feasible} if it satisfies the capacity constraint~(\ref{eq:capacity}) and a flow conservation constraint similar to~(\ref{eq:conservation}):
% \begin{equation}\label{eq:feasibleflow}
% 	0 \le f(e) \le c(e) \quad \forall e \in E,
% \end{equation}
\begin{equation}\label{eq:flow}
 \sum_{(u,v) \in E} f(u,v) - \sum_{(v,u) \in E} f(v,u) =  0 \quad \forall v \in V\setminus \{s,t\}
\end{equation}
The \emph{value} $|f|$ of a flow $f$ is the total flow from $s$ to $t$, i.e., $|f| = \sum_{(u,t) \in E} f(u,t) = \sum_{(s,u) \in E} f(s,u)$. A feasible flow $f$ with maximum value is called a \emph{maximum flow} in $N$.
For a given flow value $b$ the \emph{min-cost $s$-$t$ flow problem} is to find a feasible flow $f$ of value $|f| = b$ such that the cost $c_N(f) = \sum_{e \in E} a(e) \cdot f(e)$ is minimized.

In order to transform the graph $G=(V,E)$ of a power grid into an $s$-$t$ flow network $N$, we first add a new source vertex $s$ and a new sink vertex $t$ to $V$. 
Each generator $g \in V_G$ is connected by a directed edge $(s,g)$ with capacity $c(s,g) = x_g$ to the source $s$. 
Each consumer $u \in V_C$ is connected by a directed edge $(u,t)$ with capacity $c(u,t) = d_u$ to the sink $t$. 
Further, we replace each original undirected edge $\{u,v\} \in E$ by two directed copies $(u,v)$ and $(v,u)$, whose capacities $c(u,v) = c(v,u)$ are given by their common thermal limit $c(\{u,v\})$. 

Next, we define the edge costs. 
It is well known that a convex, piecewise linear edge cost function $h$ can easily be modeled in a flow network by replacing the respective edge $(u,v)$ with as many copies as the linear pieces of the cost function. 
The edge capacities are defined by the differences between consecutive breakpoints of $h$ and sum up to $c(u,v)$; the individual costs correspond to the costs as defined by the linear pieces between the breakpoints. 
Thus for ease of presentation we refrain from explicitly modeling convex piecewise linear cost functions in $N$. 
We rather assume that the flow cost $z_\lambda$ is given for each edge $(u,v)$ ($u \ne s$, $v \ne t$) by the weighted loss function $z_\lambda((u,v),f) = (1-\lambda) \cdot \ell_{\{u,v\}}(f(u,v))$, where $\lambda$ is the weight parameter of Equation~(\ref{eq:objective}). 
The edges $(s,g)$ from the source $s$ to a generator $g$ have cost $z_\lambda((s,g),f) = \lambda \cdot \gamma_g(f(s,g))$ and the edges incident to the sink $t$ have cost $0$. 
Then the objective function to be minimized is $z_\lambda(f) = \sum_{e \in E} z_\lambda(e,f)$.
Finally, we set the target flow value $b$ to the total demand $\sum_{u \in V_C} d_u$ of all consumers. 
By construction, every feasible minimum-cost flow in $N$ is a feasible minimum-cost flow in the underlying power grid $G$ and vice versa.
\section{Placing Flow Control Buses} 
\label{sec:hybridmodel}
In this section we seek to answer the question how many flow control buses
are necessary to obtain a globally optimal solution.  
Recall that the flow model is a relaxation of the physical model and uses fewer 
constraints. Therefore, optimal solutions in the flow model are at least as good as
in the physical model.

Given a power grid $G=(V,E)$, we say that making the vertices in $F$ flow control buses achieves \emph{full control} if the objective value of an
optimal energy flow for the grid $G_F$ is the same as the objective
value of an optimal solution in the flow model (or equivalently in the
hybrid network $G_V$, where every vertex is a flow control bus).  Our experiments indicate that in the IEEE instances often a
small fraction of the vertices is sufficient to achieve full control.
Afterwards we give a graph-theoretical explanation of this behavior.
\subsection{Experiments}
\label{sub:experiments-control-units}
\begin{figure}[tb!]%
\centering
  \begin{subfigure}[t]{.492\textwidth}
  	\centering
    	\includegraphics[width=\plotscale\linewidth, page=4,trim=0cm 0cm 0cm 0cm]{figures/plotCostsVsLosses.pdf}
	\caption{Total costs and losses for the IEEE benchmark data with 30 buses, where the square cross marks the solution computed by OPF.}
	\label{fig:plot-costs-losses}
\end{subfigure}
\hfill
  \begin{subfigure}[t]{.492\textwidth}
  	\centering
%  		trim = left, bottom, right, top
    	\includegraphics[width=\plotscale\linewidth, page=1,trim=0cm 0cm 0cm 0cm]{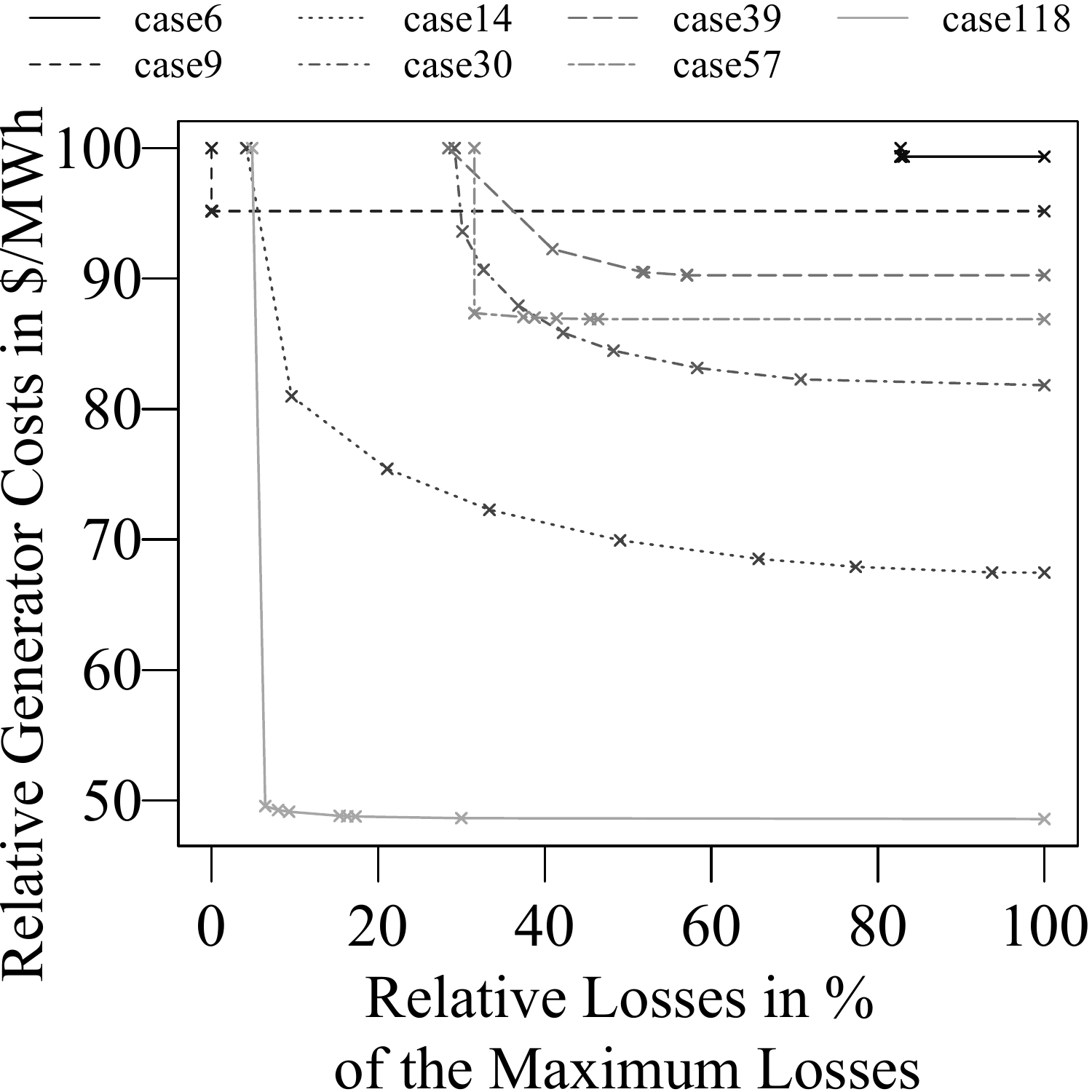}
	\caption{Relative costs and losses for the IEEE benchmark data sets.}
	\label{fig:plot-costs-losses-normalized}
\end{subfigure}
\caption{Trade-off of generator costs and losses
          normalized to the maximum generator cost ($\lambda = 0$) and
          the maximum loss ($\lambda = 1)$ as $\lambda$ varies
          from $0$ to $1$.}
\end{figure}

For our evaluation we use the IEEE benchmark data sets\footnote{data
  sources \url{http://www.pserc.cornell.edu/matpower/} and
  \url{http://www.ee.washington.edu/research/pstca/}} shown in
Table~\ref{tab:examples}. There each case is named according to the
number of buses $nb$.  The number of generators and the number of
edges are denoted $ng$ and $nl$, respectively.

 To obtain piece-wise linear functions for generator costs and line
 losses, we simply sample the cost functions using a specified number
 of sampling points.  Note further that our approach requires convex
 cost functions, but this is fine in practice~\cite{wood1996power}; in
 particular the functions are convex for the IEEE benchmark instances.

 % >cat /proc/cpuinfo
 We performed our experiments on an AMD Opteron processor 6172 running
 openSUSE 12.2. Our implementation is written in Python 2.7.3 and uses
 PYPOWER\footnote{\url{https://pypi.python.org/pypi/PYPOWER/4.0.0}}, a
 Python port of MATPOWER~\cite{Zimmerman2009,Zimmerman2011a}, for
 computing OPF solutions.  For computing solutions and minimizing the
 number of control buses in our hybrid model we use the (integer)
 linear programming solver Gurobi
 6.0.0\footnote{\url{www.gurobi.com}}.

 First, we observe that the value of $\lambda$, which controls the
 weighting of costs and losses in the objective value has a
 significant effect on the objective values of generator costs and
 line losses.  
 Figure~\ref{fig:plot-costs-losses} shows the trade-off for the IEEE
 instance \texttt{case30} (the plots for the other instances can be
 found in Appendix~\ref{app:hybridmodel}).  The OPF solution, which
 ignores losses, is typically at the far end of the spectrum with high
 losses and is comparable to our solution with $\lambda = 1$.  As can
 be seen in Figure~\ref{fig:plot-costs-losses-normalized}, where the
 costs and losses are normalized to the maximum cost and the maximum
 loss per instance, the same trade-off behavior is present in all
 instances.  It thus makes sense to allow the operator of a power grid
 to choose the value of $\lambda$ in order to model the true operation
 costs.

On the other hand, it may then be the case that the number of flow control buses to achieve full control of the network varies depending on the
choice of~$\lambda$.  Figure~\ref{fig:plot-controller-weight} shows
for different values of~$\lambda$ the relative number of control
vertices necessary to achieve full control in each of the instances.
In most cases less than 15\% of all buses need to be controllers to
achieve full control.  For the cases with 6 buses and 14 buses this
percentage is slightly bigger, which is mainly an artifact stemming
from the small total size.  As can be seen, the required number of
units is relatively stable but drops to zero for~$\lambda = 1$,
i.e., when only the generator costs are considered. This is due to the fact that all IEEE instances have basically unlimited line capacities and thus do not restrict the possible flows.

\begin{figure}[tb!]%
\centering
  	%  		trim = left, bottom, right, top
    	\includegraphics[width=\plotscaleOne\linewidth, trim=0cm 0cm 0cm 0cm]{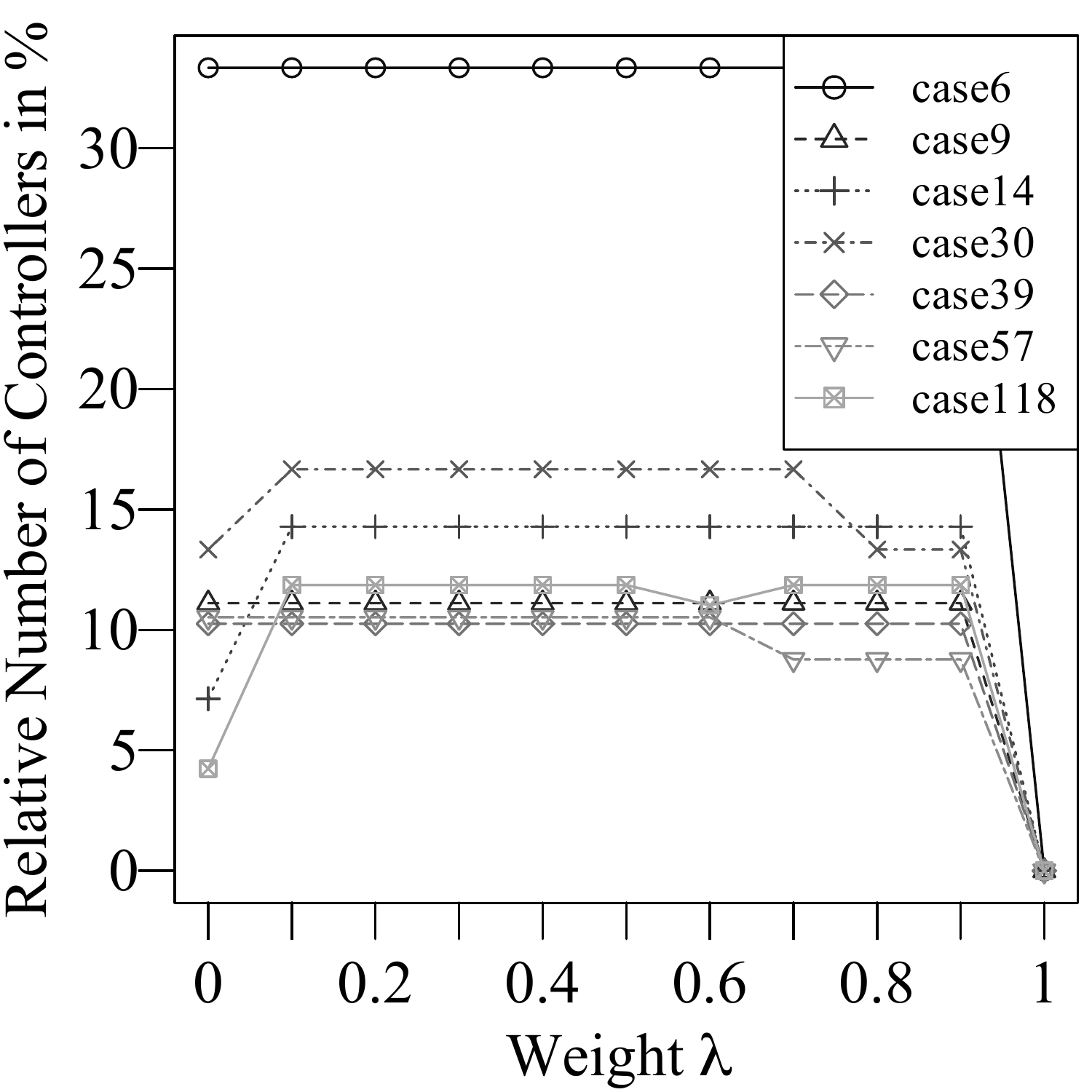}
	\caption{Relative number of controllers for achieving full
          control in the IEEE instances as $\lambda$ varies from $0$ to
          $1$.}%
	\label{fig:plot-controller-weight}%
\end{figure}%

In order to make a useful prediction on the number of vertices
required for full control that applies to all choices of~$\lambda$, in
the following we take for each instance the maximum of the smallest
possible number of vertices to achieve full control over all values 
of~$\lambda$ and refer to this as the number of vertices for achieving
full control of the instance.  This conservative choice ensures that
the numbers we compute are certainly an upper bound for achieving full
control, independent of the actual choice of~$\lambda$.
\subsection{Structure of Optimal Solutions}
\label{sub:hybridtheory}
As we have seen in our experimental evaluation, often a small number
of flow control buses is sufficient to ensure that solutions in the hybrid
model are the same as in the flow model.  In the following we provide
a theoretical explanation of this property and link it to structural
properties of power grids. Farivar and Low~\cite{6507352} give similar structural 
results on spanning trees, but using a different model.

A first observation is that flow control buses influence all
incident edges.  Thus, if every edge is incident to a flow control bus, i.e., the set $F$ is a \emph{vertex cover} of $G$, no
edge in the network is affected by
constraint~(\ref{eq:electricalflowintro}).  Then the flow model and
the hybrid model are equivalent and full control is achieved.
However, it is generally not true that power grids admit small vertex
covers; as shown in Figure~\ref{fig:barchart-cases-controller}, all
instances require more than 40\% of their vertices for a vertex cover.
In the following we show a much stronger result, namely that it
suffices for becoming independent of 
Equation~(\ref{eq:electricalflowintro}) that the native power grid
$G-F$ is an acyclic network.  Moreoever, if $\lambda = 1$, (line
losses are neglected) and edge capacities are ignored, it even
suffices that $G-F$ is a so-called \emph{cactus} graph, in which every
edge is part of at most one cycle.
\begin{figure}[tb!]
  	\centering
%  		trim = left, bottom, right, top
    \includegraphics[width=\plotscaleOne\linewidth, trim = 0cm 0cm 0cm 0cm]{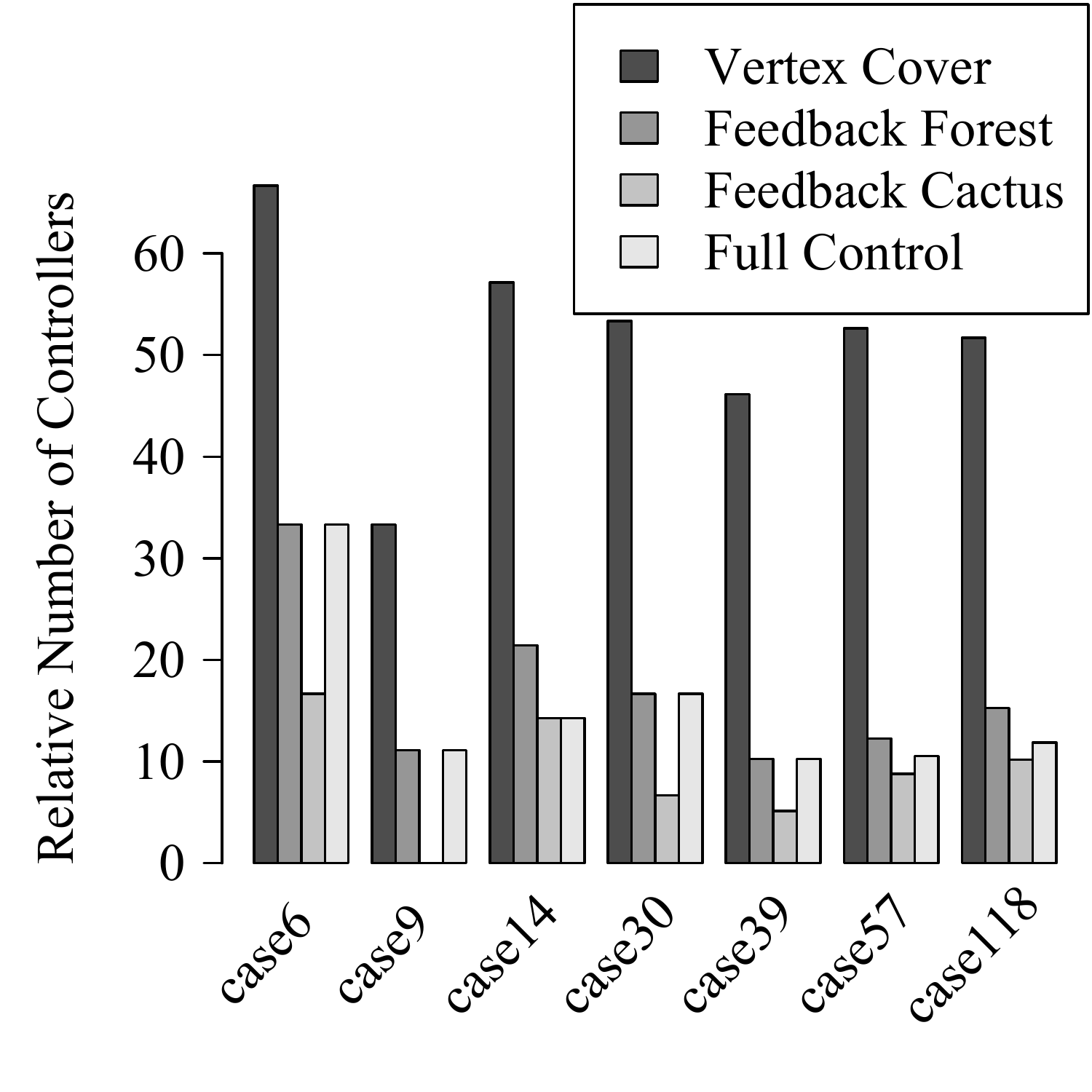}
    \caption{Comparison of the number of vertices which need to be
      removed from the network to get a Tree (Feedback Vertex Set) or
      a cactus, with the worst number of controller to have full
      control in the network.}
	\label{fig:barchart-cases-controller}
\end{figure}
\begin{lemma}
  \label{lem:decompose-powergrid-cutvertex}
  Let $H = (V,E)$ be a native power grid and let $v$ be a vertex whose
  removal disconnects $H$ into connected components with vertex sets
  $C_1,\dots,C_k$.  Then a flow $f$ is electrically feasible for $H$
  if and only if it is electrically feasible for $H_i = H[C_i \cup
  \{v\}]$ for $i=1,\dots,k$.
\end{lemma}
\begin{proof}
  Clearly, if $\Theta$ is a feasible voltage angle assignment for $H$, then its
  restriction to $C_i \cup \{v\}$ is a feasible angle assignment for $H_i$.
  Conversely, assume that $\Theta_i$ is a feasible angle assignment for
  $H_i$.  Define $\Theta_i' = \Theta_i - \Theta_i(v)$.  
  Since for
  every edge in $H_i$ the voltage angles of the endpoints are changed by the same value,
   $\Theta'$ is a feasible voltage angle assignment for $H_i$.
  Further, $\Theta_i'(v) = 0$ for every $H_i$, which means that the function $\Theta \colon V \to \mathbb{R}$, where $\Theta(u) \mapsto
    \Theta_i'(u)$ for $u \in C_i$ is well-defined.
  Note that the restriction of $\Theta$ to any of the $H_i$ coincides
  with $\Theta_i'$.  Since every edge of $H$ belongs to exactly one of
  the $H_i$, it follows that $\Theta$ is a feasible voltage angle assignment for $H$.
\end{proof}
Iteratively applying Lemma~\ref{lem:decompose-powergrid-cutvertex}
yields the following.
\begin{corollary}
  \label{cor:electrical-feasible-blocks}
  A flow in a native power grid is electrically feasible if and only
  if it is electrically feasible for each biconnected component of the
  power grid.
\end{corollary}

We observe that if $G-F$ is a forest, then each biconnected component
$H$ consists of a single edge $\{u,v\}$.  Then $\Theta(u) = f(u,v) /
B(u,v)$ and $\Theta(v) = 0$ are feasible voltage angles for any flow $f$
in $B$. Thus we conclude with the following
\begin{theorem}
  \label{thm:fvs}
  Let $H$ be a native power grid that is a forest.  Then every flow
  $f$ is electrically feasible on $H$.
\end{theorem}
Thus, when $F$ is a \emph{feedback vertex set} of $G$, i.e., $G-F$ is
a forest, then every flow on $G$ is electrically feasible for $G-F$,
and thus any feasible flow for $G_F$ is electrically feasible for
$G_F$.  It follows that the flow model and the hybrid model are
equivalent in this case.  In particular, whenever $F$ is a feedback
vertex set, instead of solving the LP for the hybrid model, we can
rather assume the flow model and compute an optimal solution using a,
potentially more efficient, flow algorithm.  It follows from
Theorem~\ref{thm:fvs} that this solution is optimal also in the hybrid
model.

Figure~\ref{fig:barchart-cases-controller} shows for each of our
instances the relative number of vertices necessary to obtain a vertex
cover, a feedback vertex set with respect to forests, and the number
of vertices necessary to obtain full control.  In all instances a
vertex cover is two to three times larger than a feedback vertex set
(for forests) and the vertex set necessary for full control.
Comparing the relative number of controllers for full control with the
size of a feedback vertex sets shows that the number to get an optimal
placement is in many cases smaller than the size of a feedback vertex
set.  Thus, in the optimal solutions, the native power grid does not
always represent a forest, but can also include cycles.  A closer
inspection showed that this is in particular the case for instances
that are operated far from their capacity limits.

We now consider what happens when cycles exist in a native power
grid.  To this end, we start with the simplest case of a power grid
that consists of a single cycle $C$.  We say that two flows $f$ and
$f'$ on a network $G=(V,E)$ are \emph{equivalent} if for each vertex
$v \in V$ we have $f_\net(v) = f_\net'(v)$.
\begin{lemma}
  \label{lem:cycle-equivalent-flow}
  Let $C$ be a native power grid that is a cycle.  For every
  flow $f$ there exists a unique equivalent flow $f'$ that is
  electrically feasible for $C$.
\end{lemma}
\begin{proof}
  Let $v_1,\dots,v_n$ be the vertices of $C$ as they occur along the
  cycle, i.e., $f(v_i,v_j)= 0$ unless $v_i$ and $v_j$ are neighbors on the cycle.
%  $i$ and $j$ differ by at most~$1$ modulo $n$.  
  Assume we wish to change the amount of flow
  from $v_1$ to $v_2$ by a fixed amount $\Delta$ and obtain an
  equivalent flow.  The net out-flow conservation at the vertices then
  uniquely determines the change of flow along the remaining edges.
  Hence, every flow $f'$ equivalent to $f$ is obtained from $f$ by
  choosing some amount $\Delta$ and setting $f'(v_i,v_{i+1}) =
  f(v_i,v_{i+1}) + \Delta$ and $f'(v_{i+1},v_i) = f(v_{i+1},v_i) -
  \Delta$, where $v_{n+1} = v_1$.

  Now the existence of a suitable offset $\Delta$ and the associated
  feasible voltage angles can be expressed as a linear system of equations.
  Namely, for edge $(v_i,v_{i+1})$, $i=1,\dots,n$ (again using
  $v_{n+1} = v_1$), we have the equation 
  \[B(v_i,v_{i+1}) \cdot
  \Theta(v_i) - B(v_i,v_{i+1}) \cdot \Theta(v_{i+1}) - \Delta =
  f(v_i,v_{i+1})\,.\] 
%
% LGS in Appendix zeigen??
\iffalse
  \begin{figure*}[t]
    \centering
    $
    \begin{pmatrix}
      B(v_1,v_2) & - B(v_1,v_2) & 0 & \dots & 0\\
      0 &  B(v_2,v_3) & -B(v_2,v_3) & \dots & 0\\
      0 &  \multicolumn{2}{c}{\dots} & \ddots & \vdots\\
      -B(v_n,v_1) & \multicolumn{2}{c}{\dots} & B(v_n,v_1)
    \end{pmatrix}
    \begin{pmatrix}
      \Theta(u_1)\\
      \Theta(u_2)\\
      \vdots\\
      \Theta(v_n)\\
      \Delta
    \end{pmatrix}
    =
    \begin{pmatrix}
      f(v_1,v_2)\\
      f(v_2,v_3)\\
      \vdots\\
      f(v_n,v_1)\\
    \end{pmatrix}
    $
  \end{figure*}
\fi
%
It is readily seen that the $n$ equations are linearly independent,
and hence a solution exists.  Moreover, dividing each of the equations
by $B(v_i,v_{i+1})$ and summing them up yields $-\sum_{i=1}^n
1/B(v_i,v_{i+1})\Delta = \sum_{i=1}^nf(v_i,v_{i+1})/B(v_i,v_{i+1})$, which shows
  that the value $\Delta$ is uniquely determined.
\end{proof}
Note however, that the equivalent flow $f'$ whose existence is
guaranteed by Lemma~\ref{lem:cycle-equivalent-flow} does not
necessarily satisfy the capacity constraints. Also the evaluation of
$f'$ in terms of line losses may change.  If neither of these is a
limiting factor, e.g., if $\lambda = 1$ and line capacities are
sufficiently large, we can show a stronger version of
Theorem~\ref{thm:fvs}.  Recall that a cactus is a graph where every edge belongs
to at most one cycle.
\begin{theorem}
  \label{thm:cactus}
  Let $G_F$ be a power grid with flow control buses at the vertices in $F$ such
  that the maximum native power grid $G-F$ is a cactus  and every edge
  of $G-F$ that lies on a cycle has infinite capacity. For any feasible flow $f$ 
  there exists an equivalent feasible flow $f'$
  that is electrically feasible for $G_F$.
\end{theorem}
\begin{proof}
  We first construct an equivalent flow $f'$ as follows.  For each
  biconnected component $C$ of $G-F$ that is a cycle, we consider the
  restriction $f_C$ of $f$ to $C$.  By
  Lemma~\ref{lem:cycle-equivalent-flow}, there exists a unique
   flow $f_C'$ equivalent to $f_C$ that is electrically
  feasible for $C$.  We now define \[f'(u,v) =
  \begin{cases}
    f'_C(u,v) & \text{if } u,v \text{ are in a cycle } C\\
    f(u,v) & \text{otherwise}
  \end{cases}
\]
Note that changing in $f$ the flow along the edges of a cycle $C$ to
the values determined by $f_C'$ preserves the net out-flow at every
vertex, and hence $f'$ is a flow equivalent to $f$.  We claim that
$f'$ is electrically feasible.  To see this, observe that each block
of $G-F$ is either a single edge or a cycle $C$.  In the former case,
$f'$ is trivially feasible on the block.  In the latter, we have that
$f'$ coincides on $C$ with $f_C'$, which is electrically feasible.  By
Corollary~\ref{cor:electrical-feasible-blocks} $f'$ is electrically
feasible.
\end{proof}
Let $e_1, \dots, e_k$ be the edges of a cycle in $G_F$ and $f_i$ be a flow on an edge $e_i$ in cycle $C$. 
We abbreviate the susceptance $B(e_i)$ on an edge in a cycle by $B_i$. The maximum susceptance is denoted by $B_{\max}$ with $B_{\max} = \max_{1\leq i\leq k}(B_i)$ for all $i = 1, \dots, k$. 
The minimum susceptance $B_{\min}$ is defined analogously.
In practice, the requirement for infinite capacity in Theorem~\ref{thm:cactus} is unnecessary. In fact, we 
can bound the sufficient large capacities of Theorem~\ref{thm:cactus} 
by rearrange the equation of the proof of Lemma~\ref{lem:cycle-equivalent-flow} 
such that the change of flow is bounded by the ratio of maximum to minimum susceptance times the average flow in the cycle $C$ that is
%\[\Delta = -\frac{\sum_{i = 1}^{n} \frac{f(v_i,v_{i+1})}{B(v_i,v_{i+1})}}{\sum_{i=1}^{n} \frac{1}{B(v_i, v_{i+1})}}\leq\frac{B_{max}}{B_{min}}\cdot \frac{\left(\sum_{i = 1}^{n} f_i\right)}{n}.\] 
\[\Delta = -\frac{\sum_{i = 1}^{k} \frac{f_i}{B_i}}{\sum_{i=1}^{k} \frac{1}{B_i}}\leq\frac{B_{max}}{B_{min}}\cdot \frac{\left(\sum_{i = 1}^{k} f_i\right)}{k}.\] 

We refer back to Figure~\ref{fig:barchart-cases-controller}, which in
addition to the previously mentioned parameters also shows the size of
a minimum feedback vertex set with respect to cacti.  In all cases the
number of vertices for full control is between the size of feedback
vertex sets with respect to forests and cacti.  For the cases 14, 57
and 118, the minimum number of controllers for achieving full control
indeed results in a native power grid that forms a cactus, although
they do not necessarily achieve the smallest feedback number due to
some influence of line capacities.
\section{Grid Operation Under Increasing Loads}	
\label{sec:grid-control-when-approaching-capacity-limits}
In the previous section we have seen that typically selecting a small
fraction of the buses as flow control buses suffices to achieve full
control in the network.  In this section we study what happens when
even fewer flow control buses are available and whether few flow control buses
allow a better utilization of the existing infra-structure in the
presence of increasing loads.% as they may occur in the future.

To measure the controllability in the presence of very few flow
control buses, we simulate a load increase by a factor~$\rho$ in the
power grid by decreasing all line capacities by the factor $1/\rho$.
This has the effect that the overall demand remains constant and thus
any change of costs is due to flow redirections.  It is then expected
that, once the load increases, the network without flow control buses
will require significantly higher operating costs, since the main
criterion for determining the generator outputs becomes the overall
feasibility of the flow rather than the cost-efficient generation of
the energy.  At some point, the load increases to a level where, by
means of changing only the generator outputs, a feasible energy flow
cannot be found.  We compare the operation costs to solutions in power
grids with a small number of flow control buses.  Specifically, our
plots show two things.  First, the operation costs for various small
numbers of flow control buses and, second, the operation costs and the
number of flow control buses for achieving full control in the network
with respect to the load increase factor~$\rho$.

Of course these operation costs again vary depending on the value of
$\lambda$.  Since most related work ignores line losses, we consider
only the case $\lambda = 1$, i.e., only generation costs are taken
into account.  Varying $\lambda$ changes the objective value, but it
does not influence the existence of solutions with a certain number of
flow control buses.  Recall from the plot in
Figure~\ref{fig:plot-controller-weight} that, if the load
increase~$\rho$ is small, full control can be achieved without flow
control buses for $\lambda=1$.
\begin{figure}[tb!]
  	\centering
    	\includegraphics[width=\plotscaleOne\linewidth, page=7]{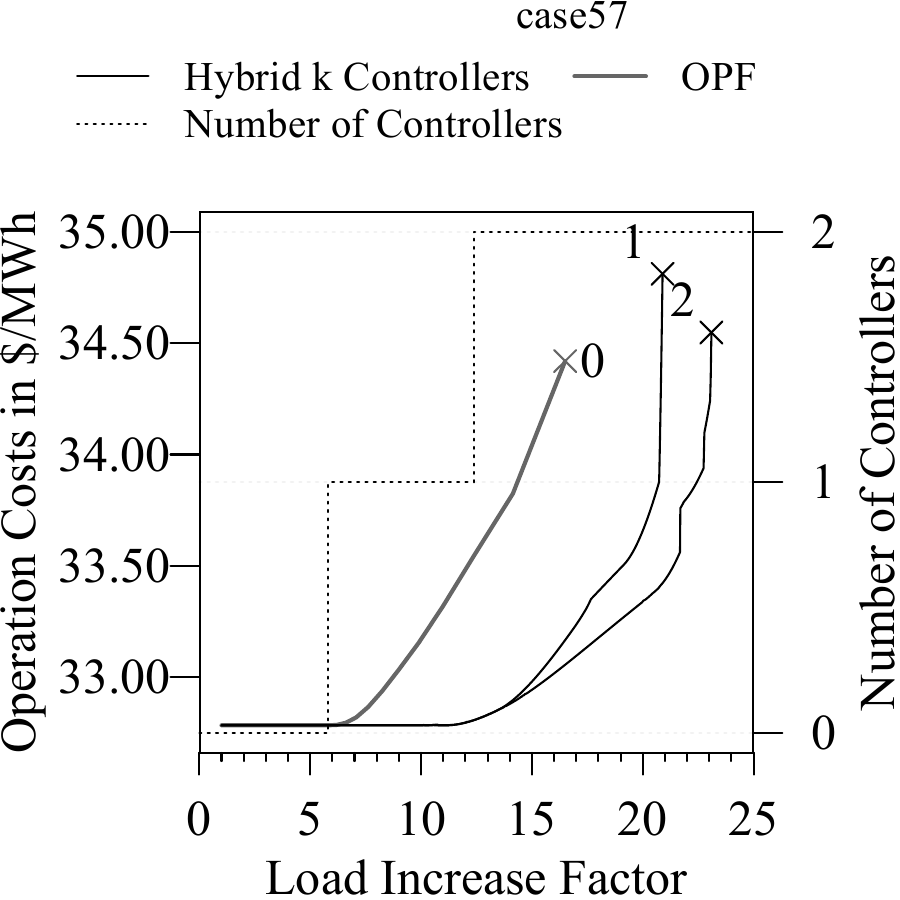}
	\caption{Operation costs of \texttt{case57} for OPF and the
          hybrid model with $1$ and $2$ control buses with respect to
          the load factor~$\rho$.}
	\label{fig:plot-capacity-cost-controller}
\end{figure}
In the IEEE instances all lines have very large capacities, often much
larger than even the total demand in the network, e.g., 9900 MW in
\texttt{case14} and \texttt{case57}.  To better highlight the
interesting parts, similar to the work by Lima et
al.~\cite{lgkm-psp-03}, we first scale all line capacities such that
the smallest capacity is equal to the total demand of the consumers as
given in Table~\ref{tab:examples}.  This changes neither the existence
nor the cost of solutions.  We increase the load until the flow model
becomes infeasible; at this point a feasible solution cannot be
achieved by adding flow control buses and adding additional lines to
the network becomes unavoidable.

Figure~\ref{fig:plot-capacity-cost-controller} shows the results of
our experiment for the power grid \texttt{case57}.  To improve
readability, all costs have been rescaled by the total demand in the
network, and thus give the cost per MWh.  The black curve shows the
operation cost with sufficient control buses for full control.  The
dotted staircase curve shows the number of flow control buses that are
necessary to achieve full control.  Moreover, for each number of flow
control buses from~1 up to the number required at the point when
further load increase makes the instance infeasible, we show the
optimal operation costs with this number of flow control buses.
Finally, the bold gray curve shows the operation cost with OPF, i.e.,
without any control buses.  The plots for the other IEEE
instances % exhibit a similar behavior and
can be found in
Appendix~\ref{app:grid-control-when-approaching-capacity-limits}.

As expected, increasing loads result in increasing operation costs.
Interestingly, very few control buses suffice for extending the
maximal feasible operation point.  This is emphasized by the curve for
two control buses in Figure~\ref{fig:plot-capacity-cost-controller},
which continues to a load increase of factor~23.09, whereas OPF works
only for up to an increase of roughly 17.27 and exhibits a significant
increase in operation costs at higher loads.  In contrast, when using
flow control buses, the costs start increasing much later and more
moderately.  Interestingly, the solution with one control bus remains roughly
equivalent to the solution with two control buses until shortly before
the end of its feasibility range.  This example shows that control buses
indeed extend the feasible operation point and also decrease the corresponding
operation costs even if there are only very few controllers available.
\section{Conclusion}	
\label{sec:conclusion}
Assuming the existence of special buses that control the flow on all
their incident transmission lines, we have presented a hybrid model
for including some flow control buses.  In this model, we have shown
that relatively few control buses suffice for achieving full control.
Further, we scaled the load of the network and showed that even fewer
flow control buses improve the loadability and have a lower cost
increase compared to OPF.

Our work shows the benefits of augmenting power grids with flow
control devices.  Using our theoretical model, we were able to explain
our empirical observations on controller placement with
graph-theoretical means.  While this also explains previous
observations of Gerbex et al.~\cite{gcg-olmtf-01}, the main drawback
is that the model is based on several strong, simplifying assumptions.

Future work should consider more realistic power grid models both in
terms of the control units, which are placed on transmission lines
rather than buses, and using the AC power grid model.
\bibliographystyle{abbrv}
\bibliography{bibliography/books,bibliography/article,bibliography/internet} 
\clearpage
\appendix
\section{Problem Formulation} \label{app:lp}
In this appendix, we present the problem formulation from Section~\ref{sec:model} as an integer linear program (LP) formulation. In the LP below, we minimize the generation costs $c_g$ and losses $c_\ell$ shown in Equation~\ref{eq:app:objective} under flow and electrical constraints. The main flow constraints comprise the conservation of flow, demand and generation constraints, and capacity constraints shown in Equations~\ref{eq:app:conservation},~\ref{eq:app:demand},~\ref{eq:app:generator} and~\ref{eq:app:capacity}, respectively. Whereas the electrical constraints describe the electrical feasibility shown in Equation~\ref{eq:app:electricalfeasible}. Recall that each consumer $u$ has a power demand $d_u\in\mathbb{R}$ and $F\subseteq V$ is the set of flow control buses.
\begin{figure*}[b!]
\centering
\begin{align}
& \underset{f}{\text{min}}
& c_\lambda(f)\,=\, & \lambda \cdot c_g(f) + (1-\lambda)\cdot c_\ell(f)\label{eq:app:objective}\\
& \text{s.t.}
& \sum_{\{v,u\} \in E} f(v,u)\,=\,&  0 & & & & v \in V\setminus (V_G \cup V_C)\label{eq:app:conservation}\\
& & f(u,v)\,=\,& B(u,v) (\Theta(u) - \Theta(v)) & & & & \forall u,v\in V\setminus F, \forall\{u,v\}\in E\label{eq:app:electricalfeasible}\\
& & \sum_{\{v,u\} \in E} f(v,u)\,=\,&   -d_v & & & & v \in V_C\label{eq:app:demand}\\
& & 0\,\le\,&\sum_{\{v,u\} \in E} f(v,u) \,\le\,x_v & & & & v \in V_G\label{eq:app:generator}\\
& & -c(e)\,\le\,& f(u,v)\,\le\, c(e) & & & &\forall e=\{u,v\} \in E.\label{eq:app:capacity}
\end{align}
\end{figure*}
\section{Placing Flow Control Buses} \label{app:hybridmodel}
\begin{figure}[H]%
  \begin{subfigure}[t]{.45\textwidth}
  	\centering
    	\includegraphics[width=0.95\linewidth, page=1,trim=0cm 0cm 0cm
        0cm]{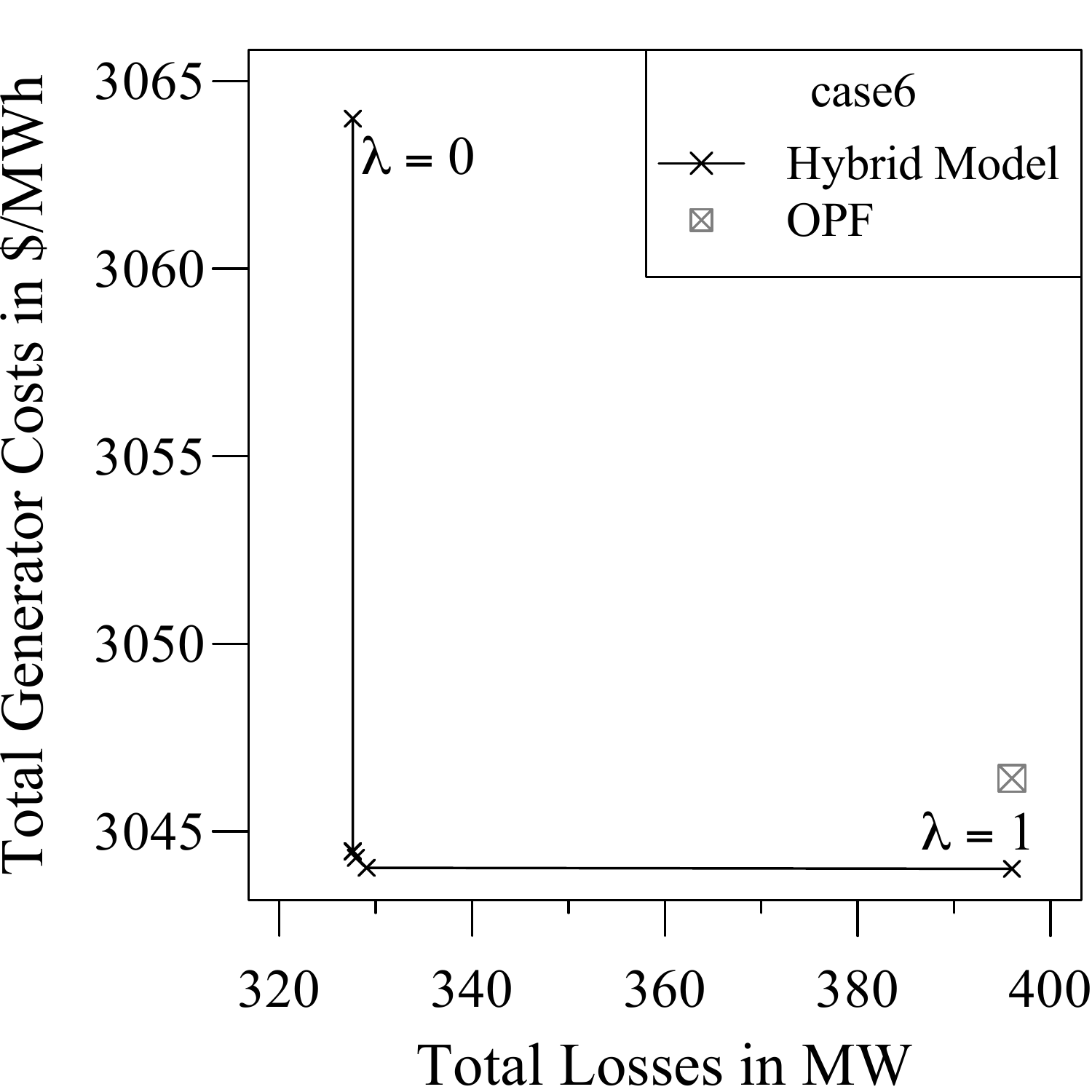}
	\label{fig:plot-costs-losses-case6}  
  \end{subfigure}
\hfill
  \begin{subfigure}[t]{.45\textwidth}
  	\centering
    	\includegraphics[width=0.95\linewidth, page=2,trim=0cm 0cm 0cm
        0cm]{figures/plotCostsVsLosses.pdf}
	\label{fig:plot-costs-losses-case9}
\end{subfigure}

\vspace{0.3cm}
\begin{subfigure}[t]{.45\textwidth}
  	\centering
    	\includegraphics[width=0.95\linewidth, page=3,trim=0cm 0cm 0cm
        0cm]{figures/plotCostsVsLosses.pdf}
\end{subfigure}
\hfill
\begin{subfigure}[t]{.45\textwidth}
  	\centering
    	\includegraphics[width=0.95\linewidth, page=5,trim=0cm 0cm 0cm 0cm]{figures/plotCostsVsLosses.pdf}
	\label{fig:plot-costs-losses-case39}
\end{subfigure}

\vspace{0.3cm}
\begin{subfigure}[t]{.45\textwidth}
  	\centering
    	\includegraphics[width=0.95\linewidth, page=6,trim=0cm 0cm 0cm 0cm]{figures/plotCostsVsLosses.pdf}
	\label{fig:plot-costs-losses-case57}
\end{subfigure}
\hfill
\begin{subfigure}[b]{.45\textwidth}
  	\centering
    	\includegraphics[width=0.95\linewidth, page=7,trim=0cm 0cm 0cm 0cm]{figures/plotCostsVsLosses.pdf}
	\label{fig:plot-costs-losses-case118}
\end{subfigure}
\vspace{0cm}
\caption{Trade-off of generator costs and costs of the losses
          depending as $\lambda$ varies from $0$ to $1$.   The square cross marks the solution computed by OPF.}
\end{figure}
\section{Grid Operation Under Increasing Loads}
\label{app:grid-control-when-approaching-capacity-limits}
\begin{figure}[H]
\vspace{-0.4cm}
  \begin{subfigure}[t]{.45\textwidth}
  	\centering
    	\includegraphics[width=0.95\linewidth, page=3]{figures/plotCapacityReductionVsCostsController.pdf}
%	\caption{case6.}
	\label{fig:plot-capacity-cost-controller-case6}
\end{subfigure}
\hfill
\begin{subfigure}[t]{.45\textwidth}
  	\centering
    	\includegraphics[width=0.95\linewidth,
        page=4]{figures/plotCapacityReductionVsCostsController.pdf}
%	\caption{case9.}
	\label{fig:plot-capacity-cost-controller-case9}
\end{subfigure}

\vspace{0.5cm}
\begin{subfigure}[t]{.45\textwidth}
  	\centering
    	\includegraphics[width=0.95\linewidth, page=5]{figures/plotCapacityReductionVsCostsController.pdf}
%	\caption{case14.}
	\label{fig:plot-capacity-cost-controller-case14}
\end{subfigure}
\hfill
\begin{subfigure}[t]{.45\textwidth}
  	\centering
    	\includegraphics[width=0.95\linewidth, page=1]{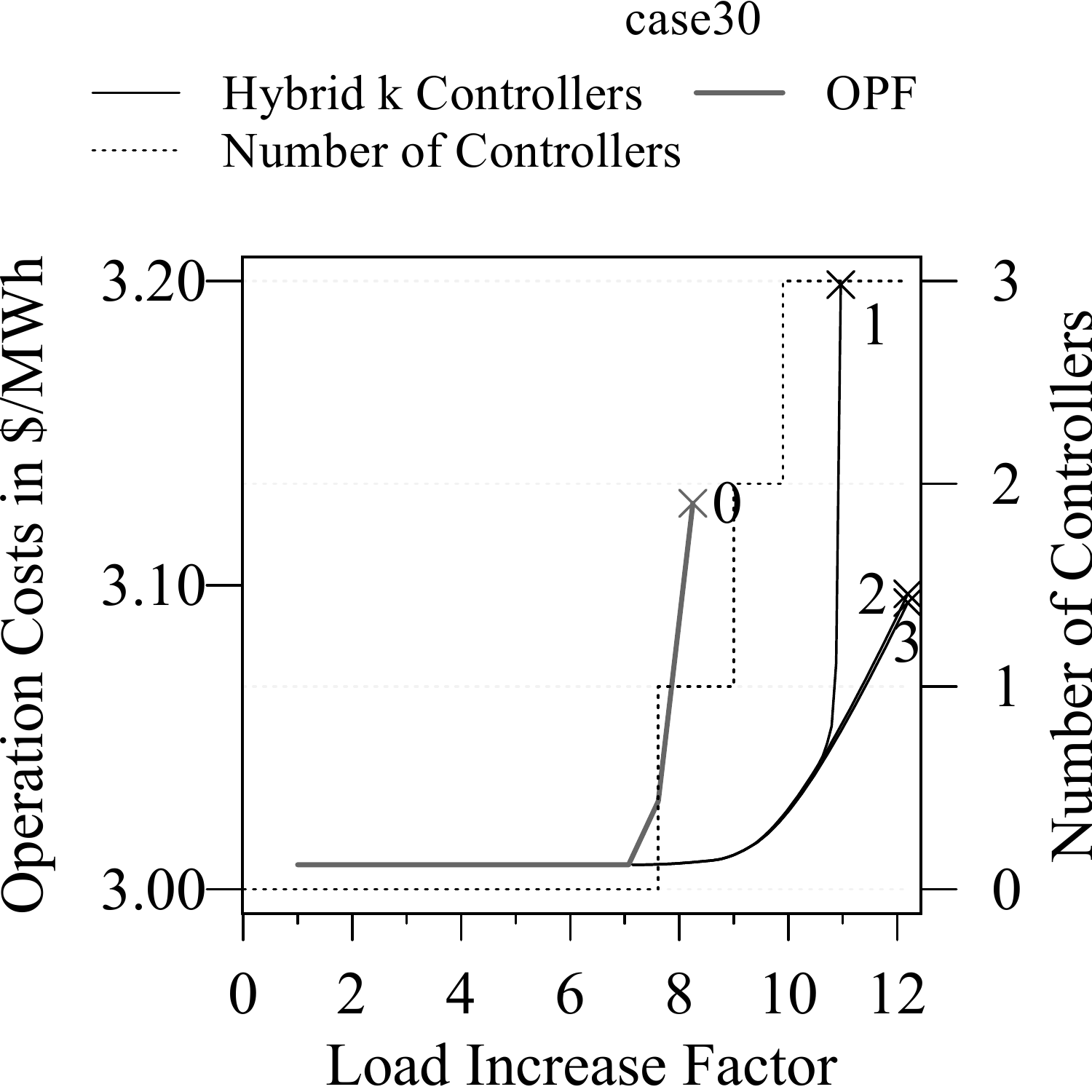}
%	\caption{case30.}
	\label{fig:plot-capacity-cost-controller-case30}
\end{subfigure}

\vspace{0.5cm}
\begin{subfigure}[t]{.45\textwidth}
  	\centering
    	\includegraphics[width=0.95\linewidth, page=2]{figures/plotCapacityReductionVsCostsController.pdf}
%	\caption{case39.}
	\label{fig:plot-capacity-cost-controller-case39}
\end{subfigure}
\hfill
\begin{subfigure}[t]{.45\textwidth}
  	\centering
    	\includegraphics[width=0.95\linewidth, page=6]{figures/plotCapacityReductionVsCostsController.pdf}
	\label{fig:plot-capacity-cost-controller-case118}
\end{subfigure}
\vspace{0cm}
	\caption{Operation costs of \texttt{case6} to \texttt{case118} for OPF and the
          hybrid model with their control buses with respect to
          the load factor~$\rho$.}
\end{figure}
\end{document}